\documentclass[letterpaper,journal]{IEEEtran}
\usepackage{amsmath,amsfonts,amssymb}
\usepackage[ruled,vlined]{algorithm2e}
\usepackage{array}
\usepackage[caption=false,font=normalsize,labelfont=sf,textfont=sf]{subfig}
\usepackage{textcomp}
\usepackage{stfloats}
\usepackage{url}
\usepackage{verbatim}
\usepackage{graphicx}
\usepackage{amsthm}
\usepackage{cite}
\newtheorem{theorem}{Theorem}

\theoremstyle{definition}
\newtheorem{definition}[theorem]{Definition}

\theoremstyle{remark}

\begin{document}
\title{Distributionally Robust Linearly Constrained Minimum Variance Beamforming Under Steering-Vector Mismatch}

\author{Raneem Madani, Abdel Lisser, and Zeno Toffano%
\thanks{Corresponding author: Raneem Madani.}%
\thanks{This work was supported by the French government through the France 2030 investment plan, under the QuanTEdu-France project (grant no. ANR-22-CMAS-0001) and the France 2030 program (reference ANR-11-IDEX-0003) within the OI H-Code.}%
\thanks{This work has been submitted to the IEEE for possible publication.
Copyright may be transferred without notice, after which this version may no
longer be accessible.}%
\thanks{Raneem Madani, Abdel Lisser, and Zeno Toffano are with the Laboratoire des Signaux et Systèmes (L2S), CNRS, CentraleSupélec, Université Paris-Saclay, 91190 Gif-sur-Yvette, France.}%
\thanks{E-mail: \{raneem.madani; abdel.lisser; zeno.toffano\}@centralesupelec.fr}}
\markboth{IEEE Transactions on Signal Processing,~Vol.~XX, No.~XX, 2026}%
{Madani \MakeLowercase{\textit{et al.}}: Distributionally Robust LCMV Beamforming}
\maketitle
\begin{abstract}
Linearly constrained minimum variance (LCMV) beamforming is a fundamental technique for controlling the array response toward multiple desired signals while suppressing interference and noise. In practical systems, however, the steering vectors involved in the response constraints are often affected by direction-of-arrival errors, calibration imperfections, and other array uncertainties. Such mismatches may cause constraint violations and substantial performance degradation. This paper investigates multi-signal LCMV beamforming under steering-vector mismatch and develops probabilistic constraints to guarantee a prescribed reliability level for each desired signal. Several uncertainty-information settings are considered. When the exact mismatch distribution is known and belongs to the class of complex elliptically symmetric distributions (CES), a distribution-informed convex safe reformulation is derived using exact CES quantiles. When the mismatch distribution is unknown and only partial statistical information or empirical data is available, robust probabilistic reformulations are developed under moment-based ambiguity, support-based uncertainty, and Wasserstein data-driven ambiguity, and resulting deterministic formulations can be expressed as convex second-order cone programs. Numerical experiments show that the proposed designs improve out-of-sample reliability and reduce beam-gain sensitivity to steering-vector mismatch, while maintaining competitive output signal to interference plus noise (SINR) ratio and robust beampatterns.
\end{abstract}
\section{Introduction}

\IEEEPARstart{A}{daptive} beamforming is a fundamental technique in array signal processing, with applications in radar, sonar, wireless communications, microphone arrays, medical imaging, radio astronomy, and spatial filtering. Among classical adaptive beamformers, the linearly constrained minimum variance (LCMV) beamformer is particularly important because it minimizes the array output interference-plus-noise power while enforcing prescribed responses in one or several look directions. The LCMV principle originates from linearly constrained adaptive array processing, in which linear equality constraints are imposed to preserve the desired responses while the remaining degrees of freedom are adapted to suppress interference and noise~\cite{frost1972algorithm}. The minimum variance distortionless response (MVDR) beamformer is a special case of LCMV in which a single distortionless-response constraint is imposed in the desired look direction~\cite{1449208}. In practice, however, the performance of LCMV and MVDR-type beamformers strongly depends on the accuracy of the steering vectors used in the constraints. Direction-of-arrival errors, array calibration imperfections, sensor position perturbations, mutual coupling, local scattering, and propagation uncertainties may produce steering-vector mismatch. Even small mismatches can lead to signal self-nulling, violations of the desired-response constraints, and severe degradation of the output signal-to-interference-plus-noise ratio (SINR).

A large body of work has therefore been devoted to robust adaptive beamforming~\cite{1165054}. Classical robustness mechanisms include diagonal loading, eigenspace-based methods, covariance regularization, steering-vector estimation, and deterministic worst-case optimization. Worst-case robust beamforming has been particularly influential: the steering-vector mismatch is assumed to belong to a deterministic uncertainty set, and the beamformer is designed to guarantee performance for every admissible mismatch realization~\cite{vorobyov2003robust,kim2008robust}. Convex optimization has played a central role in making such robust designs tractable, notably through second-order cone programming (SOCP) and semidefinite programming (SDP) reformulations~\cite{5447076}. More recent studies have extended worst-case SINR maximization to broader and potentially nonconvex uncertainty sets~\cite{huang2023robust}. Although these approaches provide strong deterministic guarantees, they may become conservative when the mismatch is random and the worst-case realization has a very low probability of occurrence.

To reduce this conservatism, probabilistic robust beamforming requires the desired response to be maintained with a prescribed probability rather than for every possible mismatch realization~\cite{vorobyov2005probability,vorobyov2008relationship}. This approach relates the robustness level to an admissible outage probability and is therefore more flexible than purely worst-case designs. Other studies model the steering-vector mismatch as a random vector and derive tractable deterministic approximations under Gaussian or related stochastic assumptions~\cite{554254,liao2016robust}.

Recent work on complex-valued chance-constrained programs (3CPs) has established tractable deterministic reformulations for linear probabilistic constraints involving the real part of uncertain complex affine expressions~\cite{10.1007/978-3-032-13589-6}. Under a known probability distribution, these reformulations generally require a detailed statistical characterization of the uncertainty, including its mean, covariance, and pseudo-covariance. Distributionally robust extensions have also been developed using partial statistical information~\cite{madani2026distributionallyrobustcomplexchanceconstrained}, as well as empirical data through Wasserstein ambiguity sets~\cite{madani2026wasserstein}. 

Among data-driven ambiguity sets, Wasserstein balls are attractive because of their statistical motivation, out-of-sample guarantees, and strong dual representations~\cite{MohajerinEsfahani2018,8814677}. In contrast, \(\phi\)-divergence ambiguity sets~\cite{bayraksan2015data} characterize distributional uncertainty by reweighting the empirical samples. This is appropriate when the observed support is trusted but is less flexible when unseen perturbations must be represented. Wasserstein ambiguity sets address this limitation by allowing probability mass to move toward nearby realizations, making them particularly suitable for data-driven problems involving perturbations of complex-valued observations.

Despite these advances, one main gap remains: existing probability-constrained robust beamforming methods predominantly address single-response MVDR formulations and do not provide a multi-signal LCMV framework with an explicit reliability requirement for each desired response. The main contributions of this paper are summarized as follows:
\begin{itemize}
\item We formulate a probabilistic multi-signal LCMV beamforming problem under steering-vector mismatch. An individual chance constraint is imposed for each desired signal to guarantee its prescribed reliability level.

\item We derive a tractable deterministic reformulation when the steering-vector mismatch follows a known complex elliptically symmetric distribution, a broad family encompassing important models such as the circular complex Gaussian, complex Student-\(t\), generalized Gaussian, and Laplace distributions.

\item We develop distributionally robust reformulations under moment-based and support-based ambiguity sets. The moment-based formulation uses available first- and second-order statistical information without assuming a specific probability distribution, while the support-based formulation exploits known bounds on the steering-vector mismatch.

\item We introduce a data-driven distributionally robust LCMV formulation based on Wasserstein ambiguity sets, allowing the mismatch distribution to be learned from samples rather than assumed to be exactly known. A worst-case conditional value-at-risk approximation is employed to obtain a tractable conservative reformulation of the resulting distributionally robust chance constraints.

\item Numerical experiments show that the proposed designs improve out-of-sample reliability and reduce sensitivity to steering-vector mismatch while maintaining competitive output SINR. They further demonstrate that, even with a limited number of samples, the Wasserstein data-driven formulation achieves performance close to that obtained when the exact mismatch distribution is known.

\end{itemize}

The remainder of this paper is organized as follows. Section~\ref{sec2} presents the system model and the nominal LCMV formulation for multiple desired signals. Section~\ref{sec3} develops the probabilistic and distributionally robust formulations under exact distributional, moment-based, support-based, and Wasserstein uncertainty descriptions, together with their tractable reformulations and solution procedures. Section~\ref{sec4} reports the numerical results. Section~\ref{sec5} concludes the paper.

\textit{Notations:}
We denote by \(\mathbb{C}^n\) the space of complex-valued vectors of dimension \(n\). For a complex vector \(z\in\mathbb{C}^n\), \(\bar{z}\), \(z^T\), and \(z^H\) denote its conjugate, transpose, and Hermitian transpose, respectively. The real and imaginary parts of a complex scalar or vector are denoted by \(\operatorname{Re}(\cdot)\) and \(\operatorname{Im}(\cdot)\), respectively. The modulus of a complex scalar is denoted by \(|\cdot|\), while the Euclidean norm of a vector is denoted by \(\|\cdot\|_2\). The identity matrix of size \(M\) is denoted by \(I_M\). For any positive integer \(N\), we define \([N]=\{1,\ldots,N\}\). The expectation and probability operators are denoted by \(\mathbb{E}[\cdot]\) and \(\mathbb{P}[\cdot]\), respectively. For a random vector \(\xi\), \(\mu_\xi=\mathbb{E}[\xi]\) denotes its mean, and \(\operatorname{Cov}(\xi)=\mathbb{E}\left[(\xi-\mu_\xi)(\xi-\mu_\xi)^H\right]\) denotes its covariance matrix. For a Hermitian matrix \(A\), \(A\succeq 0\) indicates that \(A\) is positive semidefinite.
\section{System Model and LCMV Formulation}
\label{sec2}
Consider a narrowband array with $M$ sensors. At time instant $k$, the received snapshot is denoted by $x(k)\in\mathbb{C}^{M}$. The array output is
\[
    y(k)=w^H x(k),
\]
where $w\in\mathbb{C}^{M}$ is the beamforming vector. Assume that $L$ desired narrowband signals impinge on the array from directions $\{\theta_\ell\}_{\ell\in[L]}$. The corresponding nominal steering vectors are
\[
    \hat a_\ell=a(\theta_\ell)\in\mathbb{C}^{M},\qquad \ell\in[L].
\]
The received signal is modeled as
\[
    x(k)=\sum_{\ell\in[L]}s_\ell(k)a_\ell+i(k)+n(k),
\]
where $s_\ell(k)$ is the $\ell$th desired source waveform, $a_\ell$ is the corresponding actual steering vector, $i(k)$ denotes the aggregate interference component, and $n(k)$ denotes additive noise. The received covariance matrix is
\[
    R_x=\mathbb{E}\left[x(k)x^H(k)\right].
\]
In practice, $R_x$ is estimated from $K$ snapshots as
\[
    \hat R_x=\frac{1}{K}\sum_{k\in[K]}x(k)x^H(k).
\]
The beamformer output power is
\[
    \mathbb{E}\left[|y(k)|^2\right]=w^H R_x w,
\]
or $w^H\hat R_xw$ when the sample covariance matrix is used. Let
\[
    C=\begin{bmatrix}
    \hat a_1 & \hat a_2 & \cdots & \hat a_L
    \end{bmatrix}\in\mathbb{C}^{M\times L}
\]
be the nominal constraint matrix, and let $f=\mathbf{1}_L$ be the desired response vector. To evaluate the signal preservation and interference suppression capability of a beamformer, we use the output signal-to-interference-plus-noise ratio (SINR). Let $R_s$ denote the desired-signal covariance matrix and let $R_{i+n}$ denote the interference-plus-noise covariance matrix. Then, the output SINR is defined as
\[
    {\rm SINR}_{\rm out}(w) = \frac{w^H R_s w}{w^H R_{i+n} w}.
\]
In the multi-signal case, if the desired sources are mutually uncorrelated with powers $\{\sigma_{s,\ell}^2\}_{\ell\in[L]}$, then
\[
    R_s= \sum_{\ell\in[L]}\sigma_{s,\ell}^2 a_\ell a_\ell^H .
\]
Similarly, $R_{i+n}$ collects the covariance contributions of the interference and additive noise components. The direct SINR maximization problem can be written as
\[
    \max_{w\in\mathbb{C}^{M}}\frac{w^H R_s w}{w^H R_{i+n}w}.
\]
The LCMV beamformer is closely related to this problem. Instead of maximizing the fractional objective directly, LCMV fixes the desired array responses and minimizes the output power. The LCMV problem is
\begin{equation}
\begin{aligned}
    \min_{w\in\mathbb{C}^{M}}\quad & w^H\hat R_x w\\
    \text{s.t.}\quad & C^H w=f .
\end{aligned}
\label{eq:nominal_lcmv}
\end{equation}
Equivalently, $\hat a_\ell^H w=1$ for all $\ell\in[L]$. If the desired-signal covariance matrix can be written as
\[
R_s=C R_d C^H,
\]
where $R_d$ is the covariance matrix of the desired source waveforms, then under the LCMV constraint $C^H w=f$, we have
\[
    w^H R_s w= w^H C R_d C^H w= f^H R_d f,
\]
which is independent of $w$. Consequently, maximizing the output SINR is equivalent to minimizing the interference-plus-noise power $w^H R_{i+n}w$ under the response constraints. Since $R_x=R_s+R_{i+n}$ and the desired-signal contribution is fixed by the LCMV constraints, minimizing $w^H R_xw$ is equivalent to minimizing $w^H R_{i+n}w$. Thus, the LCMV beamformer can be interpreted as a constrained SINR-maximizing design.

When $\hat R_x$ is positive definite and $C$ has full column rank, the closed-form LCMV solution is
\[
    w_{\rm LCMV}=\hat R_x^{-1}C\left(C^H\hat R_x^{-1}C\right)^{-1}f .
\]

\section{Probability-Constrained LCMV Beamforming}\label{sec3}

In practical systems, the nominal steering vector $\hat a_\ell$ may differ from the actual steering vector $a_\ell$ due to direction-of-arrival estimation errors. We model this mismatch as
\[
    a_\ell=\hat a_\ell+\delta_\ell,\qquad \ell\in[L],
\]
where $\delta_\ell\in\mathbb{C}^{M}$ is a random steering-vector mismatch. Therefore, the nominal constraint $\hat a_\ell^H w=1$ does not guarantee that the actual response $|a_\ell^H w|$ remains above the desired level. Following robust beamforming formulations \cite{vorobyov2008relationship}, we impose
\begin{equation}
\begin{aligned}
    \min_{w\in\mathbb{C}^{M}}\quad
    & w^H\hat R_x w\\
    \text{s.t.}\quad
    &
    \mathbb{P}\left[\left|(\hat a_\ell+\delta_\ell)^H w\right|\ge 1 \right]\ge p_\ell, \qquad \ell\in[L].
\end{aligned}
\label{eq:chance_lcmv}
\end{equation}
where $p_\ell\in(0,1)$ is the prescribed reliability level. Using the reverse triangle inequality, we have
\[
    \left|(\hat a_\ell+\delta_\ell)^H w\right|\ge \left|\hat a_\ell^H w\right| - \left|\delta_\ell^H w\right|.
\]
Therefore, the event $\left|\delta_\ell^H w\right| \le \left|\hat a_\ell^H w\right|-1$ implies $\left|(\hat a_\ell+\delta_\ell)^H w\right|\ge 1$. Using the reverse triangle inequality together with $|z|\ge \operatorname{Re}(z)$ for any $z\in\mathbb{C}$, we obtain
\begin{equation}
\left|(\hat a_\ell+\delta_\ell)^H w\right|
\ge \left|\hat a_\ell^H w\right|-\left|\delta_\ell^H w\right|
\ge \operatorname{Re}(\hat a_\ell^H w)-\left|\delta_\ell^H w\right|
\end{equation}
Therefore, a sufficient condition for \eqref{eq:chance_lcmv} is
\begin{equation}
    \mathbb{P}\left[\left|\delta_\ell^H w\right|\le\operatorname{Re}(\hat a_\ell^H w)-1\right]\ge p_\ell, \qquad \ell\in[L].
\label{eq:phase_safe_chance}
\end{equation}
The following subsections provide deterministic safe reformulations of \eqref{eq:phase_safe_chance} under different assumptions on the steering-vector mismatch.

\subsection{Exact CES Mismatch}
We start with the first case, where the exact distribution of the steering vector is known and follows one of the complex elliptical symmetric distributions.

\begin{definition}[CES Distribution {\cite{1502990,6263313}}]
A random vector \(\delta_\ell\in\mathbb{C}^M\) is said to follow a complex elliptically symmetric distribution with location parameter \(\mu_\ell\), covariance matrix \(\Gamma_\ell\), and characteristic generator \(\psi\), written as \(\delta_\ell\sim\operatorname{CES}(\mu_\ell,\Gamma_\ell,\psi)\), if its characteristic function is of the form
\begin{align*}
\Xi_\delta(c)=\exp\!\big(i\,\mathrm{Re}(c^H\mu_\ell)\big)\psi\!\left(c^H\Gamma_\ell c\right).
\end{align*}
\end{definition}
\begin{theorem}
Let $\delta_\ell\sim\operatorname{CES}(\mu_\ell,\Gamma_\ell,\psi)$, and 
let \(k_p=Q_{|\xi|}(p)\), where \(\xi\sim\operatorname{CES}(0,1,\psi)\), and \(Q_{|\xi|}(p)\) denotes the \(p\)-quantile of \(|\xi|\). Then the LCMV formulation can be reformulated as the following convex constraint 
\begin{align*}
    &\min_{w\in\mathbb{C}^{M}}\quad  w^H\hat R_x w\\
    &\text{s.t.}\quad
    |\mu_\ell^H w|+ k_{p_\ell}\left\|\Gamma_\ell^{1/2}w\right\|_2 \le\operatorname{Re}(\hat a_\ell^H w)-1,\quad \ell\in[L].
\end{align*}
\end{theorem}

\begin{proof}
Since $\delta_\ell\sim\operatorname{CES}(\mu_\ell,\Gamma_\ell,\psi)$, CES distributions are closed under deterministic linear transformations, for any fixed \(w\in\mathbb{C}^M\), $w^H(\delta_\ell-\mu_\ell)\sim\operatorname{CES} \left(0,w^H\Gamma_\ell w,\psi \right)$. The normalized scalar random variable
\[
\xi_\ell=\frac{w^H(\delta_\ell-\mu_\ell)}{\sqrt{w^H\Gamma_\ell w}}
\]
satisfies $\xi_\ell\sim\operatorname{CES}(0,1,\psi)$. Since $\sqrt{w^H\Gamma_\ell w}=\left\|\Gamma_\ell^{1/2}w\right\|_2$, we obtain
\[
|(\delta_\ell-\mu_\ell)^H w|=\left\|\Gamma_\ell^{1/2}w\right\|_2|\xi_\ell|.
\]
Using the definition $k_{p_\ell}=Q_{|\xi|}(p_\ell), \xi\sim\operatorname{CES}(0,1,\psi)$, it follows that
\[
\mathbb{P}\left[|(\delta_\ell-\mu_\ell)^H w|\le k_{p_\ell}\left\|\Gamma_\ell^{1/2}w\right\|_2\right]=\mathbb{P}\left[|\xi_\ell|\le k_{p_\ell}\right]\nonumber\ge p_\ell.
\]
Using $\delta_\ell^H w=\mu_\ell^H w+(\delta_\ell-\mu_\ell)^H w$, and the triangle inequality, we obtain
\[
|\delta_\ell^H w|\le|\mu_\ell^H w|+|(\delta_\ell-\mu_\ell)^H w|.
\]
Consequently,
\[
\mathbb{P}\Big[|\delta_\ell^H w|\le|\mu_\ell^H w|+k_{p_\ell}\left\|\Gamma_\ell^{1/2}w\right\|_2\Big]\ge p_\ell.
\]
Therefore, the chance constraint is safely enforced if
\[
|\mu_\ell^H w|+k_{p_\ell}\left\|\Gamma_\ell^{1/2}w\right\|_2 \le\operatorname{Re}(\hat a_\ell^H w)-1.
\]
\end{proof}

\begin{table}[t!]
\centering
\caption{Exact quantiles of $|\xi|$, where
$\xi\sim\operatorname{CES}(0,1,\psi)$.}
\label{tab:norm_quantiles_ces}
\resizebox{\columnwidth}{!}{%
\begin{tabular}{|l|l|l|}
\hline
Distribution
& Distribution of $|\xi|^2$
& $Q_{|\xi|}(p)$
\\
\hline
Gaussian
& $\operatorname{Gam}(1,1)$
& $\sqrt{-\log(1-p)}$
\\
\hline
Student-$t$, $\nu>2$
& $\dfrac{\nu-2}{\nu}F_{2,\nu}$
& $\displaystyle\sqrt{\frac{\nu-2}{2}\left[(1-p)^{-2/\nu}-1\right]}$
\\
\hline
Generalized Gaussian
& $G^{1/s}$, $G\sim\operatorname{Gam}\left(\dfrac{1}{s},b\right)$
& $\displaystyle\left(Q_{\operatorname{Gam}(1/s,b)}(p)\right)^{1/(2s)}$
\\
\hline
Laplace
& $G^2$, $G\sim\operatorname{Gam}(2,b)$
& $\displaystyle Q_{\operatorname{Gam}(2,b)}(p)$
\\
\hline
$W$-distribution
& $G^{1/s}$, $G\sim\operatorname{Gam}(1,b)$
& $\displaystyle\left[-b\log(1-p)\right]^{1/(2s)}$
\\
\hline
\end{tabular}%
}
\end{table}

The distribution of the squared amplitude
$|\xi|^2=R^2$ can be computed explicitly in several important cases. Table~\ref{tab:norm_quantiles_ces} lists univariate covariance-normalized special cases derived from the CES models summarized in \cite{6263313}. Here, $\nu>2$ denotes the degrees of freedom of the complex Student-$t$ distribution, $s>0$ denotes the exponent parameter, and $b>0$ denotes the scale parameter. We write $\operatorname{Gam}(\alpha,\varpi)$ for the Gamma distribution with shape $\alpha$ and scale $\varpi$, and $F_{d_1,d_2}$ for the Fisher distribution with degrees of freedom $d_1$ and $d_2$.

\subsection{Moment-Based Worst-Case Reformulation}
We now consider the case in which the exact distribution of \(\delta_\ell\) is unknown, while its first two moments are available. For each desired signal \(\ell\in[L]\), we define the moment-based ambiguity set
\[
\mathcal{D}_\ell=\left\{\mathbb{P}\;\middle|\;\mathbb{E}_{\mathbb{P}}[\delta_\ell]=\mu_\ell,\;\operatorname{Cov}_{\mathbb{P}}(\delta_\ell)=\Gamma_\ell\right\}.
\]
The corresponding distributionally robust chance constraint requires the prescribed response reliability to hold for every probability distribution in \(\mathcal{D}_\ell\), namely,
\[
\inf_{\mathbb{P}\in\mathcal{D}_\ell}\mathbb{P}\left[\left|\delta_\ell^H w\right|\le\operatorname{Re}(\hat a_\ell^H w)-1    \right]\ge p_\ell,\qquad \ell\in[L].
\]
This formulation provides a distribution-free guarantee based solely on the available mean and covariance information.
\begin{theorem}
For any distribution with mean $\mu_\ell$ and covariance $\Gamma_\ell$, the chance constraint
\[
\inf_{\mathbb{P}\in\mathcal{D}_\ell}
\mathbb{P}\left[\left|\delta_\ell^H w\right|
\le \operatorname{Re}(\hat a_\ell^H w)-1\right]\ge p_\ell
\]
is safely enforced by the following convex constraint:
\begin{align*}
\frac{1}{\sqrt{1-p_\ell}}
\sqrt{\left|\mu_\ell^H w\right|^2+\left\|\Gamma_\ell^{1/2}w\right\|_2^2}
\le \operatorname{Re}(\hat a_\ell^H w)-1.
\end{align*}
\end{theorem}
\begin{proof}
Using the second-order moment information, we obtain
\begin{align*}
\mathbb E\left[\left|\delta_\ell^H w\right|^2\right]
&=\mathbb E\left[\left|(\delta_\ell-\mu_\ell)^H w+\mu_\ell^H w\right|^2\right]\\
&=\left|\mu_\ell^H w\right|^2+\left\|\Gamma_\ell^{1/2}w\right\|_2^2.
\end{align*}
Let $t_\ell=\operatorname{Re}(\hat a_\ell^H w)-1$. By Markov's inequality, for $t_\ell>0$,
\[
\mathbb P\left[\left|\delta_\ell^H w\right|>t_\ell\right]
\le
\frac{\left|\mu_\ell^H w\right|^2+\left\|\Gamma_\ell^{1/2}w\right\|_2^2}{t_\ell^2}.
\]
Therefore, if
\[
t_\ell\ge
\frac{1}{\sqrt{1-p_\ell}}
\sqrt{\left|\mu_\ell^H w\right|^2+\left\|\Gamma_\ell^{1/2}w\right\|_2^2},
\]
then $\mathbb P\left[\left|\delta_\ell^H w\right|>t_\ell\right]\le1-p_\ell$.
Equivalently,
\[
\mathbb P\left[\left|\delta_\ell^H w\right|
\le
\operatorname{Re}(\hat a_\ell^H w)-1\right]\ge p_\ell,
\]
which proves the result.
\end{proof}
\subsection{Support-Based LCMV with Componentwise Disk Bounds}

Assume that the steering-vector mismatch satisfies $\delta_\ell=\mu_\ell+u_\ell$, where $u_\ell\in\mathbb{C}^M$ is zero-mean and its independent entries satisfy $|u_{\ell k}|\le l_{\ell k}$ almost surely. Define
\begin{align*}
&\mathcal D_\ell^{\mathrm{disk}}:=\left\{\mathbb P:\mathbb E_{\mathbb P}[u_\ell]=0,\ u_{\ell k},\ |u_{\ell k}|\le l_{\ell k},\ k\in[M]\right\}.
\end{align*}
\begin{theorem}
For any $p_\ell\in(0,1)$, the distributionally robust chance constraint
\[
\inf_{\mathbb P\in\mathcal D_\ell^{\mathrm{disk}}}
\mathbb P\left[\left|(\hat a_\ell+\delta_\ell)^Hw\right|\ge1\right]\ge p_\ell
\]
is safely enforced by the convex constraint
\begin{align*}
k_{p_\ell}\sqrt{\sum_{k=1}^{M}l_{\ell k}^2|w_k|^2}
\le \operatorname{Re}\left((\hat a_\ell+\mu_\ell)^Hw\right)-1,
\end{align*}
where $k_{p_\ell}=\sqrt{-2\log(1-p_\ell)}$.
\end{theorem}
\begin{proof}
Define the independent real random variables
\[
X_{\ell k}:=\operatorname{Re}\left(\overline{u_{\ell k}}w_k\right).
\]
Since $\mathbb E[u_{\ell k}]=0$, we have $\mathbb E[X_{\ell k}]=0$, and the support condition gives
\[
-l_{\ell k}|w_k|\le X_{\ell k}\le l_{\ell k}|w_k|.
\]
Moreover, $\operatorname{Re}(u_\ell^Hw)=\sum_{k=1}^{M}X_{\ell k}$.
Hence, Hoeffding's inequality yields, for every $t>0$,
\[
\mathbb P\left[\operatorname{Re}(u_\ell^Hw)\le-t\right]\le \exp\left(-\frac{t^2}{2\sum_{k=1}^{M}l_{\ell k}^2|w_k|^2}\right).
\]
Choosing
\[
t=\sqrt{-2\log(1-p_\ell)}\sqrt{\sum_{k=1}^{M}l_{\ell k}^2|w_k|^2}
\]
gives $\mathbb P\left[\operatorname{Re}(u_\ell^Hw)\ge-t\right]\ge p_\ell$. Therefore, the stated deterministic constraint ensures $\operatorname{Re}\left((\hat a_\ell+\delta_\ell)^Hw\right)\ge1$, with probability at least $p_\ell$. Since $\operatorname{Re}(z)\ge1$ implies $|z|\ge1$, the required chance constraint follows uniformly over all $\mathbb P\in\mathcal D_\ell^{\mathrm{disk}}$.
\end{proof}
\subsection{Data-Driven Reformulation Under Wasserstein Ambiguity}
In practical array processing, mismatch samples may be available from calibration data. We first recall the definition of the Wasserstein distance.
\begin{definition}[Wasserstein distance~\cite{xie2021distributionally}]\label{def:wass}
Let \(\mathcal{P}_\alpha(\Xi)\subseteq\mathcal{P}(\Xi)\) denote the set of Borel probability measures on \(\Xi\) with finite \(\alpha\)-th moment, where \(\alpha\in[1,\infty)\). Let \(d\) be a metric on \(\Xi\). The \(\alpha\)-Wasserstein distance between \(\mu,\nu\in\mathcal{P}_\alpha(\Xi)\) is defined as
\[
\bigl(W_\alpha(\mu,\nu)\bigr)^\alpha:=\inf_{\gamma\in\mathcal{H}(\mu,\nu)}\int_{\Xi\times\Xi}d^\alpha(\xi,\omega)\,\gamma(d\xi,d\omega),
\]
where \(\mathcal{H}(\mu,\nu)\) is the set of all probability measures on \(\Xi\times\Xi\) with marginals \(\mu\) and \(\nu\).
\end{definition}
In the sequel, we set \(\alpha=1\), take \(\Xi=\mathbb{C}^M\), and use the Euclidean ground metric \(d(\xi,\omega)=\|\xi-\omega\|_2\). Hence, the ambiguity set is defined using the \(1\)-Wasserstein distance. For each desired signal \(\ell\in[L]\), let \(\{\hat\delta_\ell^{(j)}\}_{j\in[N]}\) be \(N\) samples of the steering-vector mismatch and define the empirical distribution
\[
\widehat{\mathbb{P}}_{N,\ell}=\frac{1}{N}\sum_{j\in[N]}\mathsf{D}_{\hat\delta_\ell^{(j)}},
\]
where \(\mathsf{D}_x\) denotes the Dirac probability measure concentrated at \(x\). For any measurable set \(A\in\mathcal{B}(\mathbb{C}^M)\), where \(\mathcal{B}(\mathbb{C}^M)\) is the Borel sigma-algebra on \(\mathbb{C}^M\),
\[
\mathsf{D}_x(A)=
\begin{cases}
1, & x\in A,\\
0, & x\notin A.
\end{cases}
\]
Following Wasserstein distributionally robust chance-constrained approaches~\cite{8814677,xie2021distributionally,doi:10.1287/moor.2022.1275}, we consider the ambiguity set
\[
\mathcal{M}_{N,\ell}^{\epsilon}=\left\{\mathbb{P}\in\mathcal{P}_1(\mathbb{C}^M): W_1\left(\mathbb{P},\widehat{\mathbb{P}}_{N,\ell}\right)\leq\epsilon\right\},
\]
where \(\epsilon\geq0\) is the Wasserstein radius. 
Define
\[
F_\ell(w,\delta)=|\delta^Hw|-\operatorname{Re}(\hat a_\ell^Hw)+1.
\]
The worst-case CVaR approximation of the distributionally robust chance constraint is
\[
\inf_{t_\ell\in\mathbb{R}}\left\{\sup_{\mathbb{P}\in\mathcal{M}_{N,\ell}^{\epsilon}}\mathbb{E}_{\mathbb{P}}\left[\left(F_\ell(w,\delta)+t_\ell\right)_+\right]-t_\ell(1-p_\ell)\right\}\leq0.
\]
\begin{theorem}
A conservative convex reformulation of the Wasserstein distributionally robust LCMV problem, obtained through a worst-case CVaR approximation, where $j\in[N],\ \ell\in[L]$:
\begin{equation}
\begin{aligned}
\min_{w,\rho,\{t_\ell,s_{\ell j}\}}\quad & w^H\hat R_x w\\
\text{s.t.}\quad
& \left|(\hat\delta_\ell^{(j)})^H w\right|-\operatorname{Re}(\hat a_\ell^H w)
+1+t_\ell\leq s_{\ell j},\\
& \epsilon\rho+\frac{1}{N}\sum_{j\in[N]}s_{\ell j}\leq t_\ell(1-p_\ell),\\
& \|w\|_2\leq\rho,\quad t_\ell\geq0,\quad s_{\ell j}\geq0,
\end{aligned}
\label{eq:wasserstein_phase_lcmv}
\end{equation}
\end{theorem}
\begin{proof}
By monotonicity of the function $(\cdot)_+$ and CVaR, it is sufficient to impose the worst-case CVaR condition associated with \(F_\ell\). Following the worst-case CVaR epigraph reformulation developed in~\cite{madani2026wasserstein}, a conservative sufficient condition is the existence of \(t_\ell\in\mathbb{R}\) such that
\[
\sup_{\mathbb{P}\in\mathcal{M}_{N,\ell}^{\epsilon}}\mathbb{E}_{\mathbb{P}}\left[\left(F_\ell(w,\delta_\ell)+t_\ell\right)_+\right]\leq t_\ell(1-p_\ell).
\]
Define \(H_\ell(w,\delta)=\left(F_\ell(w,\delta)+t_\ell\right)_+\). Since \(H_\ell(w,\delta)\) is continuous and 
\[
\left|H_\ell(w,\delta_1)-H_\ell(w,\delta_2)\right|\leq\|w\|_2\,\|\delta_1-\delta_2\|_2.
\]
then \(\|w\|_2\)-Lipschitz in \(\delta\), it satisfies the linear-growth condition in \cite{8814677} which required for strong duality when \(\alpha=1\). Therefore, by the Wasserstein strong duality result in~\cite{doi:10.1287/moor.2022.1275,madani2026wasserstein}, we obtain
\begin{align*}
&\sup_{\mathbb{P}\in\mathcal{M}_{N,\ell}^{\epsilon}}\mathbb{E}_{\mathbb{P}}[H_\ell(w,\delta)]\\
&=\inf_{\rho_\ell\geq0}\left\{\epsilon\rho_\ell+\frac{1}{N}\sum_{j=1}^{N}\sup_{\delta\in\mathbb{C}^M}\left[H_\ell(w,\delta)-\rho_\ell\|\delta-\hat\delta_\ell^{(j)}\|_2\right]\right\}.
\end{align*}
For any \(\delta_1,\delta_2\in\mathbb{C}^M\),
\[
\left||\delta_1^Hw|-|\delta_2^Hw|\right|\leq |(\delta_1-\delta_2)^Hw|\leq \|\delta_1-\delta_2\|_2\|w\|_2.
\]
Since the positive-part operator is nonexpansive, \(H_\ell(w,\delta)\) is Lipschitz continuous in \(\delta\) with Lipschitz constant \(\|w\|_2\). Therefore, whenever \(\rho_\ell\geq\|w\|_2\), 
\[
\sup_{\delta\in\mathbb{C}^M}\left[H_\ell(w,\delta)-\rho_\ell\|\delta-\hat\delta_\ell^{(j)}\|_2\right]= H_\ell\left(w,\hat\delta_\ell^{(j)}\right).
\]
Introducing \(s_{\ell j}\geq0\) such that \(H_\ell(w,\hat\delta_\ell^{(j)})\leq s_{\ell j}\) is equivalent to
\[
\left|(\hat\delta_\ell^{(j)})^Hw\right|-\operatorname{Re}(\hat a_\ell^Hw)+1+t_\ell\leq s_{\ell j},\qquad s_{\ell j}\geq0.
\]
Hence, the worst-case CVaR condition is satisfied if
\[
\epsilon\rho_\ell+\frac{1}{N}\sum_{j=1}^{N}s_{\ell j}\leq t_\ell(1-p_\ell),\qquad \|w\|_2\leq\rho_\ell.
\]
Since the left-hand side of the first inequality is nonnegative and \(p_\ell\in(0,1)\), feasibility implies \(t_\ell\geq0\). Moreover, all loss functions have the same Lipschitz constant \(\|w\|_2\); thus, a common variable \(\rho\) may be used for all \(\ell\in[L]\).
\end{proof}
\section{Numerical Results}\label{sec4}
The numerical experiments were performed on a computer equipped with an 11th-generation Intel Core i7-1185G7 processor and 32 GB of RAM. The optimization problems were implemented in Python using CVXPY and solved with the MOSEK solver. All the results presented in this section correspond to the individual chance-constrained formulations.

We consider a uniform linear array composed of \(M=12\) antennas with half-wavelength spacing. Four desired signals arrive from $\theta_{\rm des}=[-60^\circ,-10^\circ,12^\circ,30^\circ]$, while three interfering signals arrive from $\theta_{\rm int}=[-30^\circ,20^\circ,60^\circ]$. The additive noise is modeled as a zero-mean complex Gaussian process with covariance matrix \(\boldsymbol{R}_n=\boldsymbol{I}_M\), and the interference-to-noise ratio of each interferer is fixed at INR\(=20\,\mathrm{dB}\). Unless stated otherwise, the sample covariance matrix is estimated using \(K=100\) training snapshots. Steering-vector uncertainty is represented by a complex Gaussian perturbation with a small nonzero mean induced by a direction-of-arrival bias. The standard deviations are drawn from \([0.08,0.5]\), while the angular biases are drawn from \([-0.5^\circ,0.5^\circ]\). A moderately inflated covariance model is used in the robust formulations to avoid underestimating the mismatch level. The prescribed marginal reliability is \(p=0.9\). The Wasserstein formulation uses \(50\) mismatch samples and radius \(\epsilon=0.01\), unless stated otherwise. Out-of-sample reliability is estimated using \(1000\) independent mismatch realizations. The reported results are averaged over \(100\) independent Monte Carlo runs. Solid curves indicate the empirical means, while the shaded regions represent one standard deviation.

\subsection{Output SINR Performance and Beampattern}
Figure~\ref{fig:sinr_individual} compares the average output SINR of the nominal LCMV beamformer with the Gaussian, moment-based, support-based, and Wasserstein robust designs.  The total input SNR is varied from \(-20\,\mathrm{dB}\) to \(20\,\mathrm{dB}\). An unconstrained SINR upper bound computed using the true second-order statistics is also included as a performance benchmark. At low input SNR, all methods exhibit similar performance because the output is primarily limited by the noise power. As the input SNR increases, the robust formulations consistently outperform the nominal LCMV beamformer. This difference becomes particularly visible at moderate and high SNR values, where steering-vector mismatch causes a significant loss for the nominal design. Among the robust approaches, the Gaussian and Wasserstein formulations provide the highest average output SINR, whereas the moment- and support-based formulations are slightly
more conservative.
\begin{figure}[!t]
    \centering
    \includegraphics[width=\linewidth]{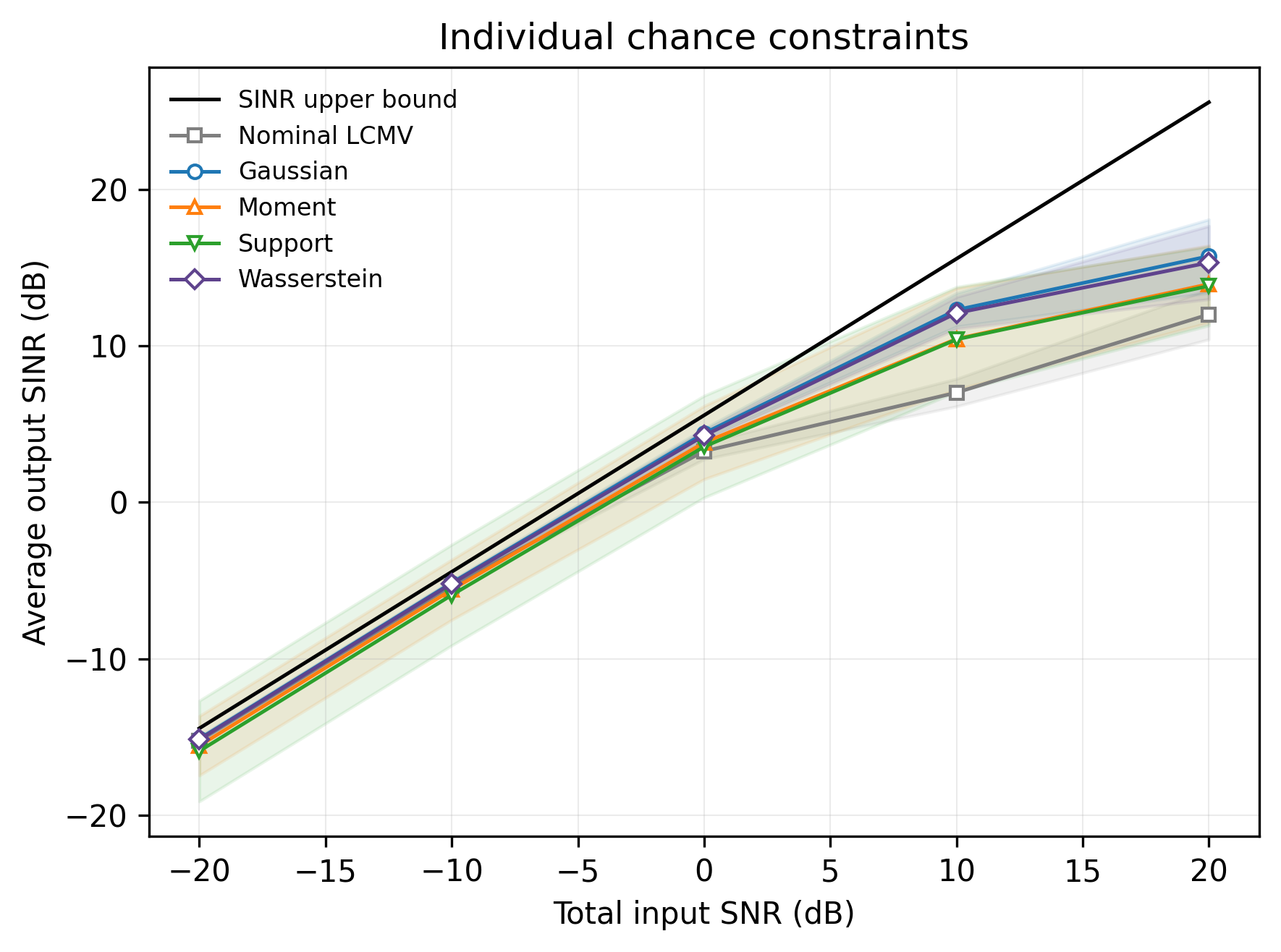}
    \caption{Average output SINR versus total input SNR.}
    \label{fig:sinr_individual}
\end{figure}

Figure~\ref{fig:sinr_snapshots} studies the effect of the number of training snapshots at a fixed total input SNR of \(10\,\mathrm{dB}\). The number of snapshots is varied from \(10\) to \(100\). The output SINR improves rapidly when the number of snapshots increases from \(10\) to approximately \(50\), since the sample covariance matrix becomes more accurately estimated. Beyond this range, the improvement becomes more gradual. The robust beamformers maintain a clear advantage over the nominal LCMV design for all considered sample sizes. In particular, the Gaussian and Wasserstein methods achieve the highest output SINR, approaching \(12\,\mathrm{dB}\) when many snapshots are available.
\begin{figure}[!t]
    \centering
    \includegraphics[width=\linewidth]{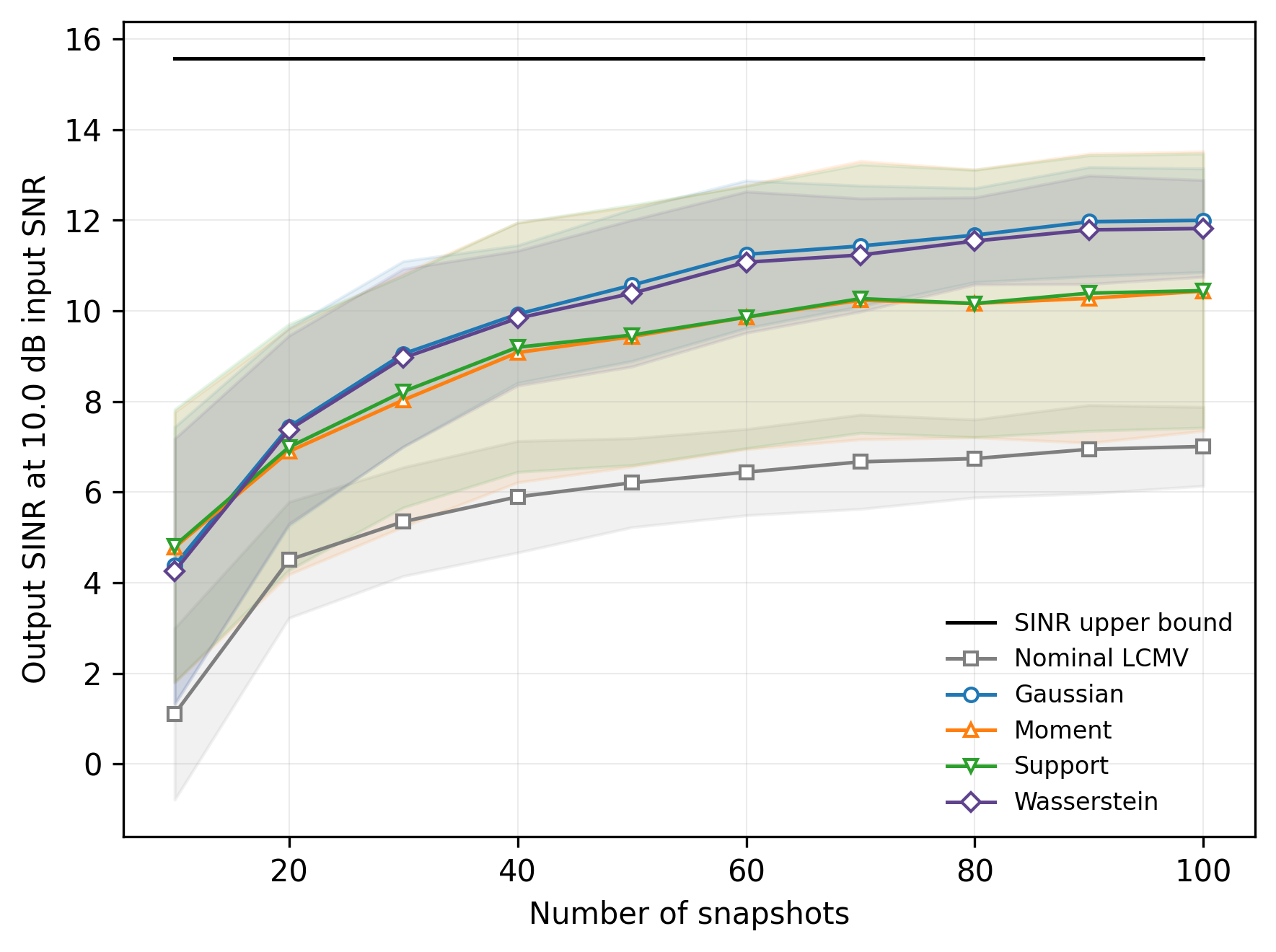}
    \caption{Average output SINR versus the number of training snapshots
    at a total input SNR of \(10\,\mathrm{dB}\).}
    \label{fig:sinr_snapshots}
\end{figure}

To illustrate the spatial filtering behavior, we evaluate the array beampattern over the angular grid \(\Theta_{\rm grid}=[-90^\circ,90^\circ]\). For a beamforming vector \(w\), the beampattern in decibels is defined as $B_{\rm dB}(\theta)= 20\log_{10}\left|a^H(\theta)w\right|$. Figure~\ref{fig:beampattern_individual} shows the average beampatterns at a total input SNR of \(20\,\mathrm{dB}\). Black dotted lines indicate the desired directions, while red dashed lines indicate the interference directions. The robust formulations provide positive response margins around the desired directions while producing strong attenuation near the interference directions. In contrast, the nominal LCMV beamformer enforces unit responses only at the presumed desired steering vectors and therefore provides less protection against steering-vector mismatch. The Gaussian and Wasserstein formulations generally produce larger response margins around the desired directions, whereas the moment- and support-based formulations exhibit more conservative spatial responses. 
\begin{figure}[!t]
    \centering
    \includegraphics[width=\linewidth]{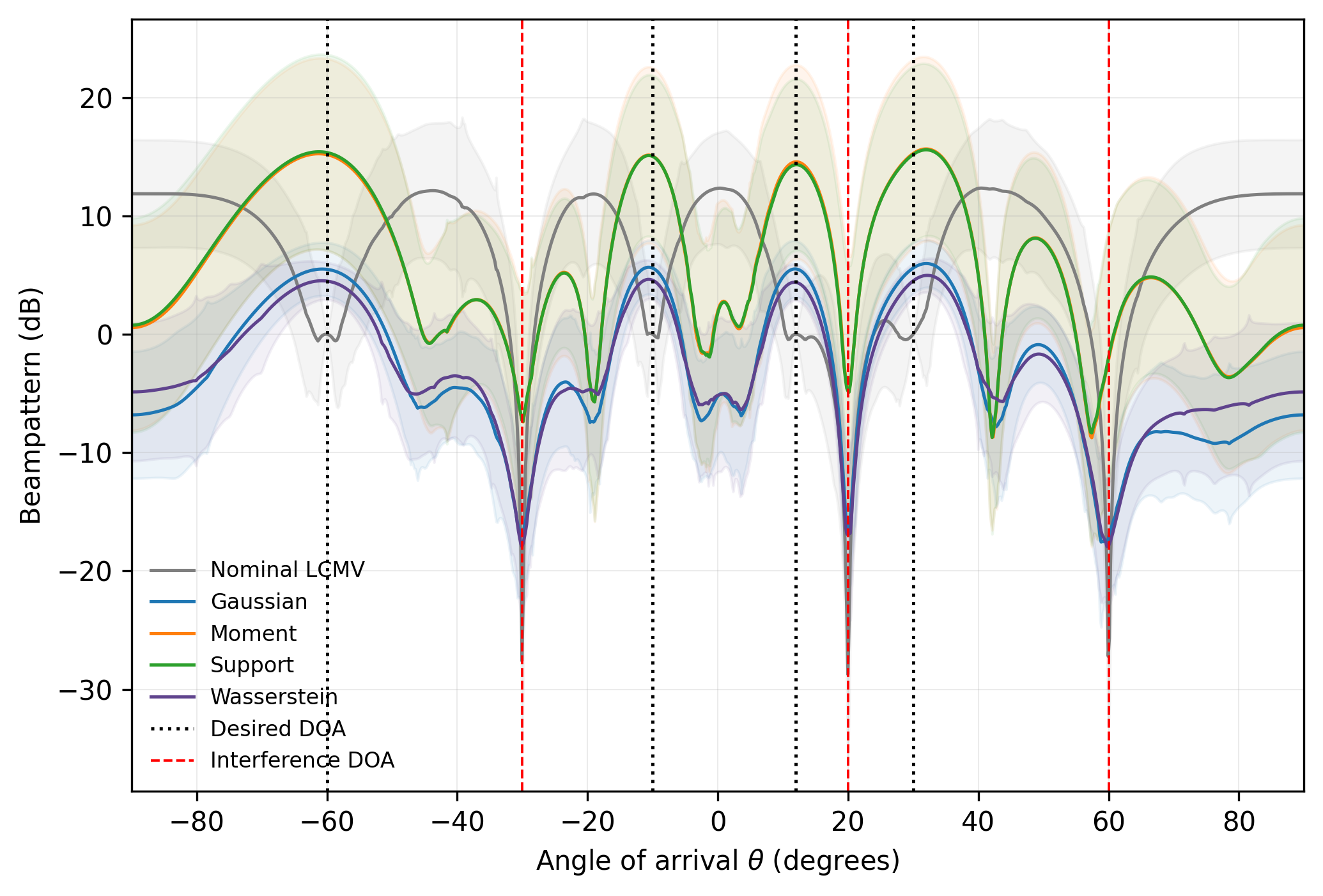}
    \caption{Average beampattern in decibels at a total input SNR of \(20\,\mathrm{dB}\).}
    \label{fig:beampattern_individual}
\end{figure}

\subsection{Reliability and Sensitivity Analysis}
Figure~\ref{fig:oos_individual} presents the out-of-sample satisfaction probabilities at a total input SNR of \(20\,\mathrm{dB}\). The nominal LCMV beamformer satisfies the response constraints in only approximately half of the out-of-sample realizations and therefore fails to attain the prescribed reliability. By contrast, all robust formulations exceed the target \(p=0.9\). Their satisfaction probabilities are close to one, showing that the proposed safe reformulations provide strong protection against steering-vector mismatch.
\begin{figure}[!t]
    \centering
    \includegraphics[width=\linewidth]{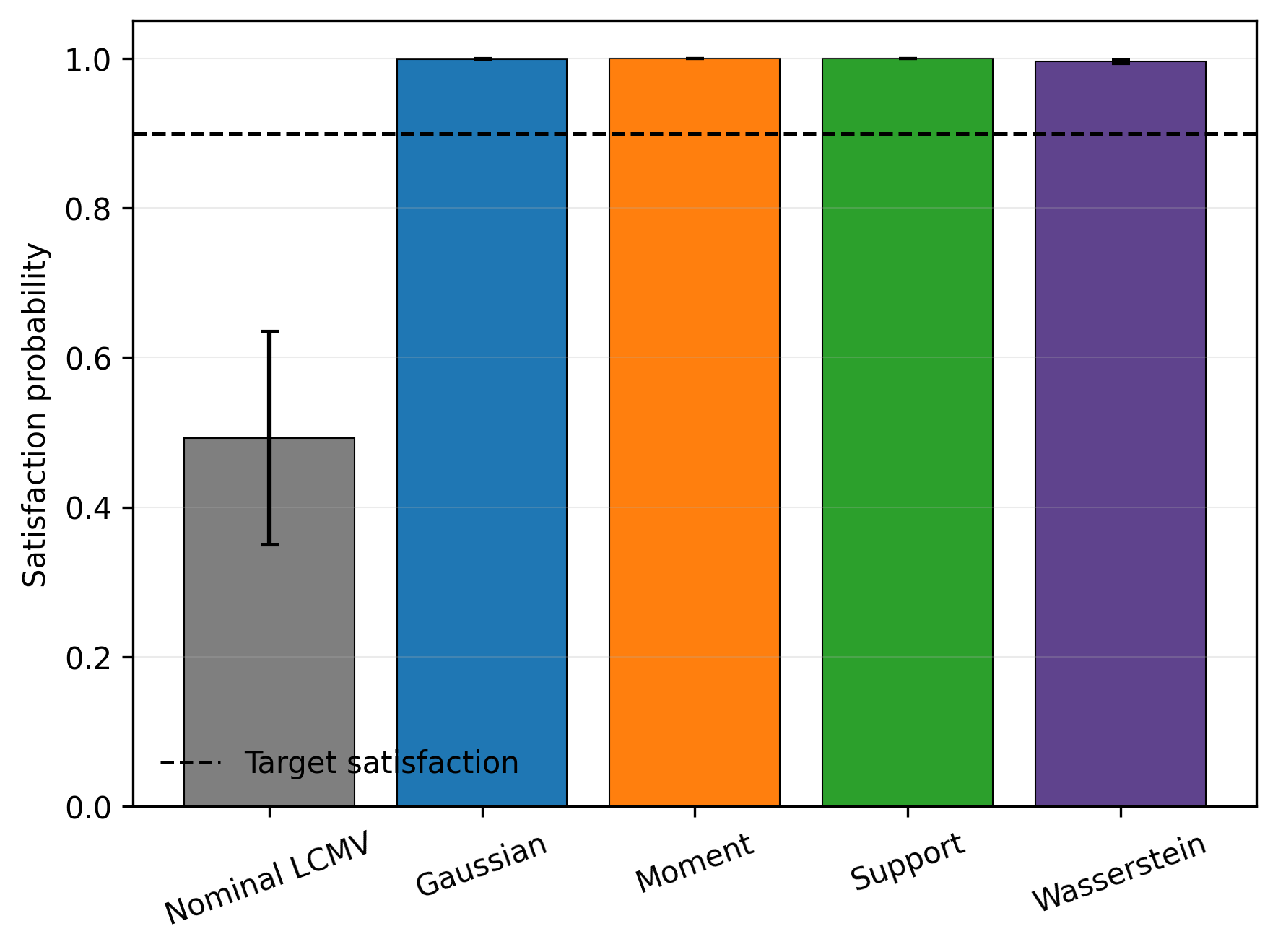}
    \caption{Minimum out-of-sample satisfaction probability at SNR \(=20\,\mathrm{dB}\). }
    \label{fig:oos_individual}
\end{figure}

Figure~\ref{fig:confidence_sweep} examines the effect of varying the prescribed reliability level \(p\) at a total input SNR of \(20\,\mathrm{dB}\). The left panel reports the average output SINR, whereas the right panel reports the corresponding minimum out-of-sample satisfaction probability.  The nominal LCMV solution is independent of \(p\), and its performance therefore remains constant throughout the experiment. The robust methods maintain an output SINR of approximately \(16\,\mathrm{dB}\) over a broad range of reliability levels. A noticeable performance reduction appears only when \(p\) approaches one, particularly for the moment, support, and Wasserstein formulations, because the corresponding constraints become increasingly restrictive. The right panel confirms that the robust solutions achieve satisfaction levels above the prescribed target. The Gaussian approximation follows the reliability requirement more closely, while the moment- and support-based formulations remain close to full satisfaction even for relatively small values of \(p\), reflecting their more conservative nature.
\begin{figure*}[!t]
    \centering
    \includegraphics[width=\linewidth]
    {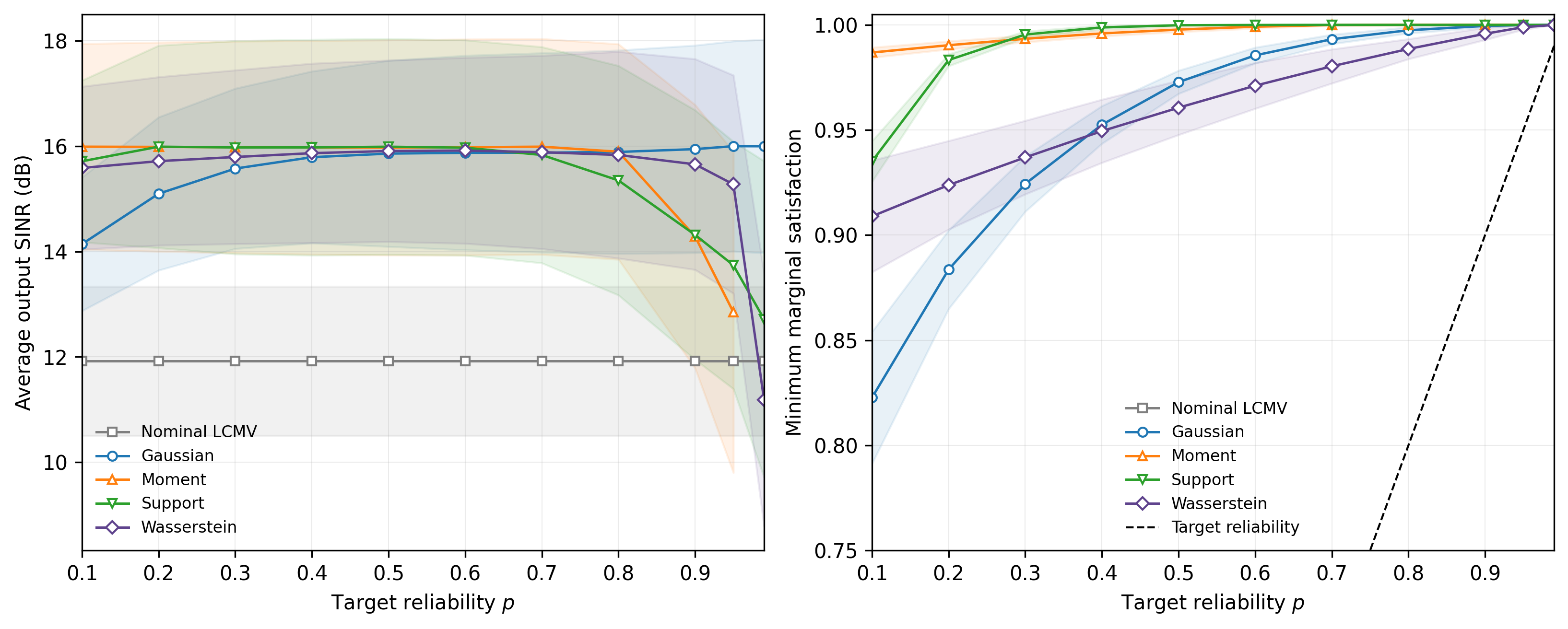}
    \caption{Effect of the prescribed reliability level at a total input SNR of \(20\,\mathrm{dB}\): average output SINR (left) and out-of-sample satisfaction probability (right). }
    \label{fig:confidence_sweep}
\end{figure*}

Finally, Figure~\ref{fig:wasserstein_radius} investigates the sensitivity of the Wasserstein formulation to the ambiguity-set radius and the number of mismatch samples. The experiment is conducted at a total input SNR of \(20\,\mathrm{dB}\), with Wasserstein sample sizes \(N\in\{25,50,100,200\}\). For small radii, the output SINR remains nearly unchanged, whereas increasing the radius to \(0.1\) produces a visible performance reduction. This behavior is expected because a larger Wasserstein ball protects against a wider set of probability distributions and therefore yields a more conservative beamformer. At the same time, the worst marginal violation probability decreases as the radius increases and becomes nearly zero for sufficiently large radii. Notably, satisfactory output SINR and reliability are already achieved with only \(N=25\) mismatch samples, while increasing the sample size provides only modest additional improvements, particularly for small ambiguity radii.
\begin{figure*}[!t]
    \centering
    \includegraphics[width=\linewidth]{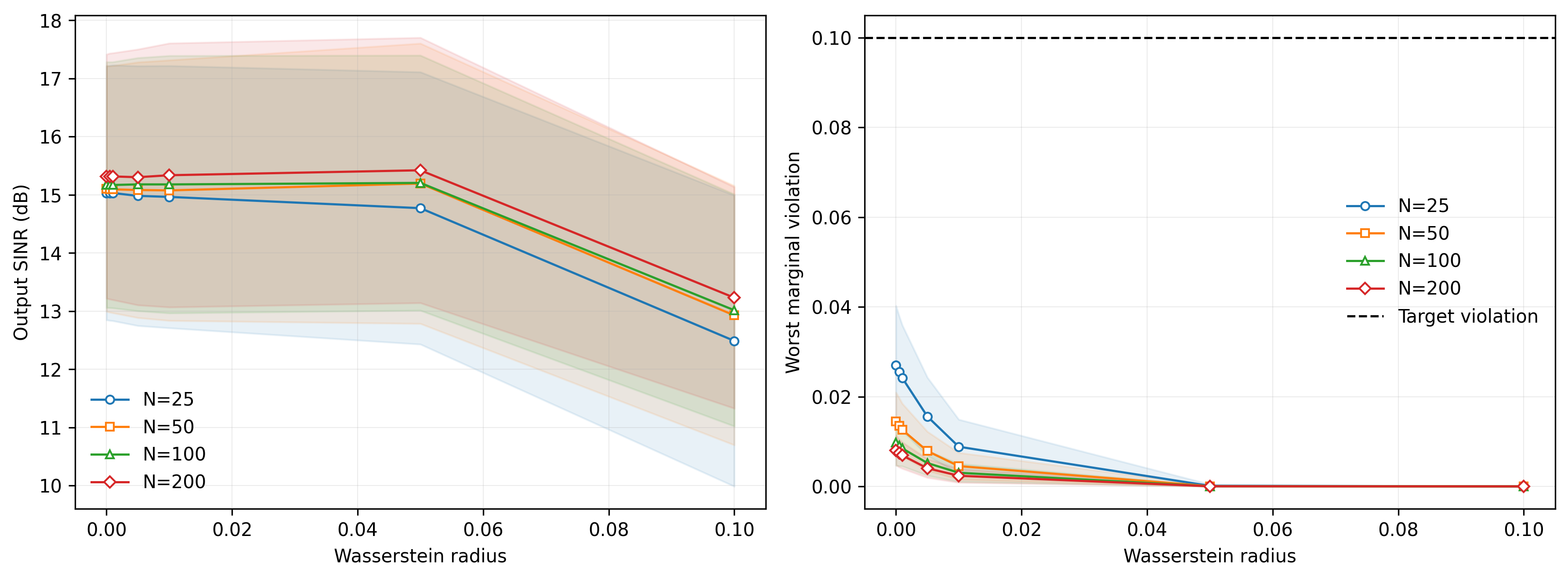}
    \caption{Effect of the Wasserstein radius and the number of mismatch samples at a total input SNR of \(20\,\mathrm{dB}\): average output SINR (left) and worst marginal violation probability (right).}
    \label{fig:wasserstein_radius}
\end{figure*}
\section{Conclusion}\label{sec5}
This paper studied robust LCMV beamforming under steering-vector mismatch using complex-valued chance constraints. In contrast to nominal LCMV designs, which enforce deterministic response constraints using presumed steering vectors, the proposed framework imposes probabilistic constraints on the uncertain array responses. This allows the desired response amplitudes to be controlled at prescribed individual reliability levels for multiple desired signals. Deterministic safe reformulations were derived under exact Gaussian information, moment-based ambiguity, support-based uncertainty, and Wasserstein data-driven ambiguity. The resulting formulations provide different levels of modeling flexibility, ranging from parametric uncertainty descriptions to fully data-driven ambiguity sets. The resulting convex programs can be solved efficiently using standard conic optimization tools. The numerical results demonstrate that steering-vector mismatch may cause substantial out-of-sample constraint violations for the nominal LCMV beamformer, even when its output SINR remains relatively high. In contrast, the proposed robust formulations attain the prescribed reliability levels while maintaining competitive output SINR. The Gaussian and Wasserstein formulations provide a favorable balance between performance and robustness, whereas the moment- and support-based models generally exhibit more conservative behavior. The experiments also show that increasing the number of training snapshots improves covariance estimation and output SINR. Increasing the reliability level $p$ or the Wasserstein radius introduces the robustness--performance trade-off. 
\bibliographystyle{IEEEtran}
\bibliography{bib}
\end{document}